\documentclass[11pt,a4paper]{article}
\usepackage{amsmath}
\usepackage{amsfonts}
\usepackage{amssymb}
\usepackage{graphicx}
\usepackage{enumerate}
\usepackage[ruled]{algorithm2e}
\usepackage{tabularx,booktabs}
\usepackage{float}

\newtheorem{theorem}{Theorem}

\newtheorem{lemma}[theorem]{Lemma}

\newtheorem{example}[theorem]{Example}

\newenvironment{proof} {\par \noindent \textbf{Proof. }}{\QED \par \bigskip \par}
\newcommand{\QED}{\hfill$\square$}

\title {
    \bf {Efficient algorithm for the vertex connectivity \\ of trapezoid graphs}
}

\author
{
{\large \sc Aleksandar Ili\' c \footnotemark[3]} \\
{\em \normalsize Faculty of Sciences and Mathematics, Vi\v segradska 33, 18000 Ni\v s, Serbia} \\
{\normalsize e-mail: { \tt aleksandari@gmail.com }}
}

\begin{document}

\maketitle

\begin{abstract}
The intersection graph of a collection of trapezoids with corner points lying on two parallel lines
is called a trapezoid graph. These graphs and their generalizations were applied in various fields,
including modeling channel routing problems in VLSI design and identifying the optimal chain of
non-overlapping fragments in bioinformatics. Using modified binary indexed tree data structure, we
design an algorithm for calculating the vertex connectivity of trapezoid graph $G$ with time
complexity $O (n \log n)$, where $n$ is the number of trapezoids. Furthermore, we establish
sufficient and necessary condition for a trapezoid graph $G$ to be bipartite and characterize trees
that can be represented as trapezoid graphs.
\end{abstract}

\vspace{0.2cm}

 \vspace{0.2cm}


\section{Introduction}

A trapezoid diagram consists of two horizontal lines and a set of trapezoids with corner points
lying on these two lines. A graph $G = (V, E)$ is a trapezoid graph when a trapezoid diagram exists
with trapezoid set $T$, such that each vertex $i \in V$ corresponds to a trapezoid $T [i]$ and an
edge exists $(i, j) \in E$ if and only if trapezoids $T [i]$ and $T [j]$ intersect within the
trapezoid diagram. A trapezoid $T [i]$ between these lines has four corner points $a [i]$, $b [i]$,
$c [i]$ and $d [i]$ -- which represent the upper left, upper right, lower left and lower right
corner points of trapezoid~$i$, respectively. No two trapezoids share a common endpoint (see Figure
\ref{fig:example}).

Trapezoid graphs were first investigated by Corneil and Kamula \cite{CoKa87}. These graphs and
their generalizations were applied in various fields, including modeling channel routing problems
in VLSI design \cite{DaGoPi88} and identifying the optimal chain of non-overlapping fragments in
bioinformatics~\cite{AbOh05}. Given some labeled terminals on the upper and lower side of a
two-sided channel, terminals with the same label will be connected in a common net. Each net can be
modeled by a trapezoid that connects rightmost and leftmost terminals of that net on two horizontal
lines. In the channel routing problem we want to connect all terminals of each net so that no two
nets intersect. One can show \cite{DaGoPi88} that two nets can be routed without intersection in
the same layer if and only if their corresponding trapezoids do not intersect. Therefore, the
number of colors needed to color the trapezoid graph is the number of layers needed to route the
nets without intersection.

Let $n$ and $m$ denote the number of vertices and edges of a trapezoid graph $G$. Ma and Spinrad
\cite{MaSp94} showed that trapezoid graphs can be recognized in $O(n^2)$ time, while Mertzios and
Corneil \cite{MeCo09} designed structural trapezoid recognition algorithm based on the vertex
splitting method in $O (n (m + n))$ time, which is easier for implementation. Trapezoid graphs are
perfect, subclass of cocomparability graphs and properly contain both interval graphs and
permutation graphs. If $a [i] = b [i]$ and $c [i] = d [i]$ then the corresponding trapezoid $T [i]$
reduces to a straight line, and the trapezoid graph reduces to a permutation graph if the condition
holds for all $1 \leq i \leq n$. Similarly, the trapezoid graph reduces to the interval graph if $a
[i] = c [i]$ and $b [i] = d [i]$ for all $i$.

\begin{figure}[ht]
  \label{fig:example}
  \center
  \includegraphics [width = 9.5cm]{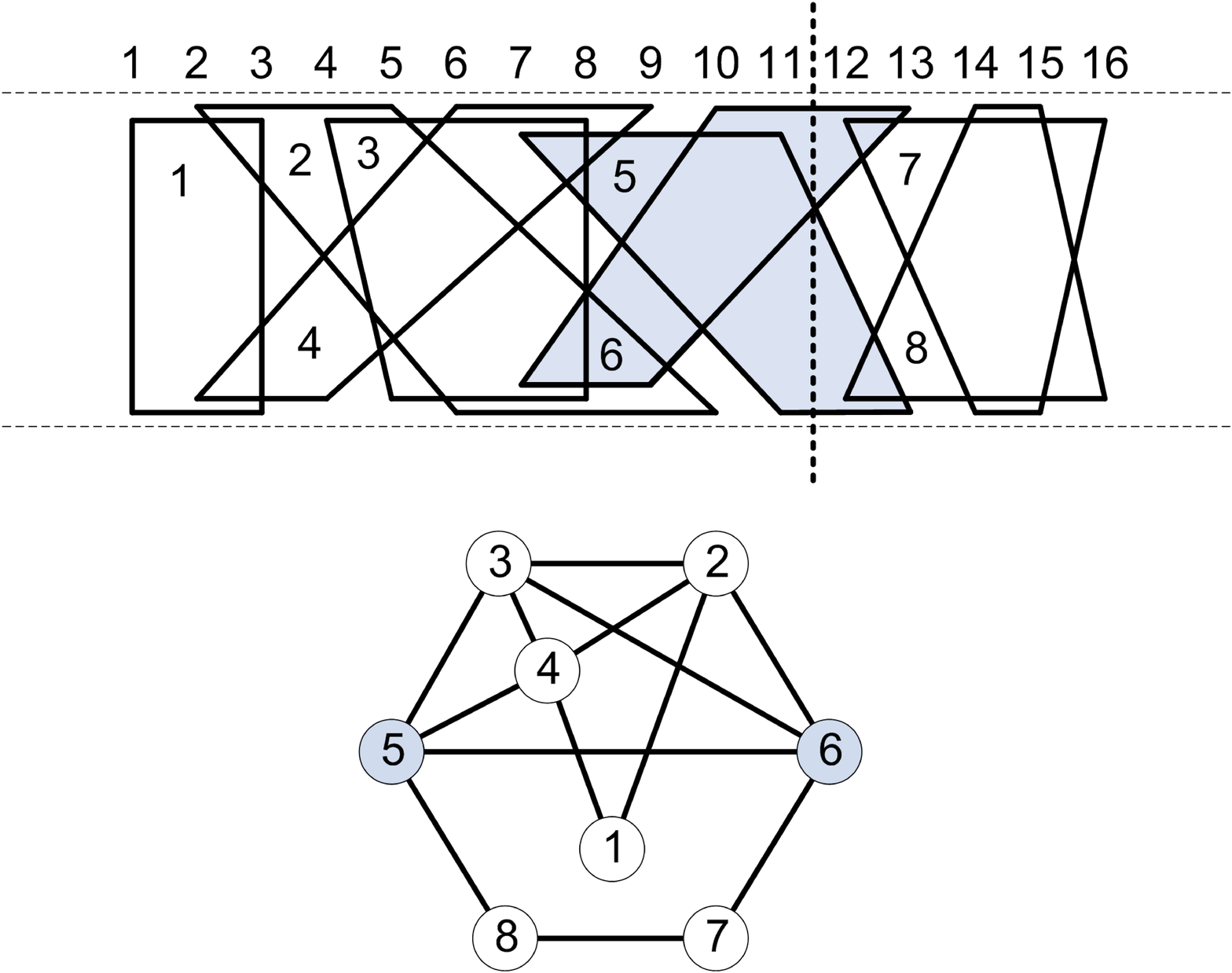}
  \caption { \textit{ A trapezoid graph and its trapezoid representation (marked vertices represent 2--cut). } }
\end{figure}

Many common graph problems, such as minimum connected dominating sets \cite{TsLiHs07}, all-pair
shortest paths \cite{MoPaPa02}, maximum weighted cliques \cite{BePaPa02}, all cut vertices
\cite{HoPaPa04}, chromatic number and clique cover \cite{FeMuWe97}, all hinge vertices
\cite{BePaPa03} in trapezoid graphs, can be solved in polynomial time. For other related problems
see \cite{ChCo96,CrGa10,Li94,Li95}. Recently, Lin and Chen \cite{LiCh09} presented $O (n^2)$
algorithms for counting the number of vertex covers (VC), minimal VCs, minimum VCs and maximum
minimal VCs in a trapezoid graph. Ili\' c and Ili\' c \cite{IlIl11} improved algorithms for
calculating the size and the number of minimum vertex covers (or independent sets), as well as the
total number of vertex covers, and reduce the time complexity to $O(n \log n)$. Ghosh and Pal
\cite{GhPa05} presented an efficient algorithm to find the maximum matching in trapezoid graphs,
which turns out to be not correct \cite{IlIl11}.

Let $G = (V, E)$ be a simple undirected graph with $|V| = n$. A vertex cut or separating set of a
connected graph $G$ is a set of vertices whose removal disconnects $G$. The connectivity or vertex
connectivity $\kappa (G)$ (where $G$ is not complete) is the size of a smallest vertex cut. A graph
is called $k$--connected or $k$--vertex--connected if its vertex connectivity is greater than or
equal to $k$. A complete graph with $n$ vertices, denoted by $K_n$, has no vertex cuts, but by
convention $\kappa (K_n) = n-1$. A connected graph is said to be separable if its vertex
connectivity is one. In that case, a vertex which disconnects the graph is called a cut-vertex or
an articulation point.

The edge cut of $G$ is a group of edges whose total removal disconnects the graph. The
edge--connectivity $\lambda (G)$ is the size of a smallest edge cut. In the simple case in which
cutting a single edge would disconnect the graph, that edge is called a bridge. Let $\delta(G)$ be
minimum vertex degree of $G$, then
$$
\kappa(G) \leq \lambda(G) \leq \delta(G).
$$

Since the strength of the network $G$ is proportional to $\kappa (G)$, graph connectivity is one of
the most fundamental problem in graph theory. Even and Tarjan \cite{EvTa75} have obtained $O(\kappa
\cdot m n \sqrt{n})$ time sequential algorithm for finding vertex connectivity of a general graph.
The authors in \cite{KaHo94} obtained parallel algorithm for testing $k$--vertex connectivity in
interval graphs. Ghosh and Pal in \cite{GhPa07} presented rather complicated algorithm with a lot
of different cases to solve the vertex connectivity problem, which takes $O(n^2)$ time and $O(n)$
space for a trapezoid graph. In this paper, we designed an algorithm with time complexity $O (n
\log n)$ for calculating the vertex connectivity of a trapezoid graph.

The rest of the paper is organized as follows. In Section 2 we introduce the modified binary
indexed tree data structure. In Section 3 we design $O (n \log n)$ time algorithm for calculating
the vertex connectivity of trapezoid graphs, improving the algorithm from \cite{GhPa07}. In Section
4 we establish sufficient and necessary condition for a trapezoid graph $G$ to be bipartite and
characterize trees that can be represented as trapezoid graphs. We close the paper in Section 5 by
proposing topics for the further research.

\section{Modified binary indexed data structure}

The binary indexed tree (BIT) is an efficient data structure for maintaining the cumulative
frequencies. We will modify this standard structure to work with minimal/maximal partial
summations.

Let $A$ be an array of $n$ elements. The modified binary indexed tree (MBIT) supports the following
basic operations:
\begin{enumerate}[($i$)]
\item for given value $x$ and index $i$, add $x$ to the element $A [i]$, $1 \leq i \leq n$;
\item for given interval $[1, i]$, find the sum of values $A [1], A [2], \ldots, A [i]$, $1 \leq i \leq n$.
\item for given interval $[1, i]$, find the minimum value among partial sums
$A [1], A [1] + A [2], A [1] + A [2] + A [3], \ldots, A [1] + A [2] + \ldots + A [i]$, $1 \leq i
\leq n$.
\end{enumerate}

Naive implementation of these operations have complexities $O (1)$, $O(n)$ and $O(n)$,
respectively. We can achieve better complexity, if we speed up the second and third operation which
will also affect the first operation.

The main idea of the modified binary indexed tree structure is that sum of elements from the
segment $[1, i]$ can be represented as sum of appropriate set of subsegments. The MBIT structure is
based on decomposition of the cumulative sums into segments and the operations to access this data
structure are based on the binary representation of the index. This way the time complexity for all
operations will be the same $O(\log n)$.

The structure is a complete binary tree with the root node 1. Its leafs correspond to the elements
from the array $A$, starting from left to right in the last level. Therefore, the elements of the
array $A$ are stored at the positions starting from $2^p$ to $2^p + n - 1$, where $p$ is the depth
of the binary tree (defined as the smallest integer such that $2^p \geq n$). The internal nodes
store the cumulative values of the leafs in the subtrees rooted at these nodes. This implies that
the value of the internal node $i$ is just the cumulative value of its two children. The parent of
the node $i$ is $\lfloor \frac{i}{2} \rfloor$, while the left and the right child of the node $i$
are $left [i] = 2i$ and $right [i] = 2i + 1$, respectively. By definition, it follows that the
number of nodes in MBIT is at most $2n$, while the depth is $\lceil \log_2 n \rceil$.

In addition to the values of the array $A$, for each node we will keep two information: $sum [x]$
will be the sum of the $A [i]$ values of all nodes in $x$'s subtree, and $min\_sum [x]$ will be the
minimum possible cumulative sum of $A [i]$'s in the subtree rooted at $x$ (starting at the leftmost
node in the subtree). We will demonstrate how to compute these fields using only information at
each node and its children. The sum of the $A[i]$'s of the subtree rooted at node $x$ will simply
be
$$
sum [x] = sum [2x] + sum [2x + 1].
$$

The minimum cumulative sum can either be in the left subtree or in the right subtree. If it is in
the left subtree, it is simply $min\_sum [left [x]]$, while if it is in the right subtree, we have
to add the cumulative sum up till the right subtree to the minimum value of the right subtree $sum
[left [x]] + min\_sum [right [x]]$. Finally, we get
$$
min\_sum [x] = \min (min\_sum [2x], sum [2x] + min\_sum [2x + 1]).
$$

For the update procedure, we just need to traverse the vertices from the leaf to the root and
update the values in the parent vertices based on the above formulas. For the query procedure, we
traverse the binary tree in a top-down manner starting from the root vertex $1$. The important
thing it to maintain the partial sums of the array $A$ (starting from $A [1]$). If the leaf that
stores $A [index]$ belongs to the left child -- we just go left and do nothing; otherwise we update
the partial sum of the left subtree and the index, and go right.

For detailed implementation see Algorithms 1 and 2. The structure is space-efficient in the sense
that it needs the same amount of storage as just a simple array of $n$ elements. Furthermore we can
use fast bitwise operations (xor, and, shift left) for more efficient implementation.

\begin{theorem}
Calculating the sum of the elements from $A [1]$ to $A [i]$, calculating the minimum value among
partial sums $A [1], A [1] + A [2], \ldots, A [1] + A [2] + \ldots + A [i]$, and updating the
element $A [i]$ in the modified binary indexed tree is performed in $O (\log n)$ time, $1 \leq i
\leq n$.
\end{theorem}

\LinesNumbered
\begin{algorithm}
    \KwIn{The value $value$, the element index $index$ and the parameters $n$ and $p = \min\{s : 2^s \geq n \}$.}

    $i = index + 2^{p} - 1$\;
    $sum [i] = value$\;
    $min\_sum [i] = value$\;
    \While {$i > 1$}
    {
        \If{$i {\ \bf mod \ } 2 = 1$}
        {
            $i = i - 1$\;
        }
        $sum [i / 2] = sum [i] + sum [i + 1]$\;
        $min\_sum [i / 2] = \min (min\_sum [i], sum [i] + min\_sum [i + 1])$\;
        $i = i / 2$\;
    }
    \caption{ Updating the modified binary indexed tree. }
\end{algorithm}

\LinesNumbered
\begin{algorithm}
    \KwIn{The index $index$.}
    \KwOut{The minimum among partial sums with elements $A[1], A [2], \ldots, A [index]$.}

    $min = \infty$\;
    $partial\_sum = 0$\;
    $i = 1$\;
    $pow = 2^{p - 1}$\;
    \While{$index > 0$}
    {
        \If{$index \geq pow$}
        {
            \If{$min > partial\_sum + min\_sum [2 \cdot i]$}
            {
                $min = partial\_sum + min\_sum [2 \cdot i]$\;
            }
            $partial\_sum = partial\_sum + sum [2 \cdot i]$\;
            $i = 2 \cdot i + 1$\;
            $index = index - pow$\;
        }
        \Else
        {
            $i = 2 \cdot i$\;
        }
        $pow = pow / 2$\;
    }
    \Return $min$\;

    \caption{ Calculating the minimal cumulative partial sum. }
\end{algorithm}

This approach is very similar to the problem regarding calculating the point of maximum overlap
among intervals (see \cite{CoLeRiSt01} Problem 14-1), that can be solved using red-black trees.

\begin{example}
Let $n = 14$ and $p = 4$. The elements of the array $A$ will be stored on positions $sum [16]$ to
$sum [29]$ and $min\_sum [16]$ to $min\_sum [29]$. The node 5 will contain the following
information
$$sum [5] = sum [20] + sum [21] + sum [22] + sum [23] = A [5] + A [6] + A [7] + A [8]$$
$$min\_sum [5] = \min (A [5], A [5] + A [6], A [5] + A [6] + A [7], A [5] + A [6] + A [7] + A
[8]).$$

If we change the value $A [6]$, we will start from the corresponding index $i = 21$ and change the
following nodes of the modified binary index tree: 21, 10, 5, 2 and 1. If we want to calculate the
minimum among partial sums with elements $A [1], A [2], \ldots, A [13]$, we calculate the following
minimum
$$
\min \left ( min\_sum [2], sum [2] + min\_sum [6], sum [2] + sum [6] + min\_sum [28] \right).
$$
\end{example}

\section{The algorithm for the vertex connectivity}

Let $T = \{1, 2, \ldots, n\}$ denote the set of trapezoids in the trapezoid graph $G = (V, E)$. For
simplicity, the trapezoid in $T$ that corresponds to vertex $i$ in $V$ is called trapezoid $T [i]$.
Without loss of generality, the points on each horizontal line of the trapezoid diagram are labeled
with distinct integers between $1$ and $2n$.

Trapezoid $i$ lies entirely to the left of trapezoid $j$, denoted by $i \ll j$, if $b [i] < a [j]$
and $d [i] < c [j]$. It follows that $\ll$ is a partial order over the trapezoid set $T$ and $(T
,\ll)$ is a strictly partially ordered set.

\begin{lemma} \cite{GhPa07}
Two vertices $T [i]$ and $T [j]$ of a trapezoid graph are not adjacent iff either ($i$) $b [i] < a
[j]$ and $d [i] < c [j]$ or ($ii$) $b [j] < a [i]$ and $d [j] < c [i]$.
\end{lemma}

Define a cut line as line $p$ that passes through the intervals $(x, x + 1)$ and $(y, y + 1)$ on
the upper and the bottom horizontal line, respectively, and does not contain the integer points $x,
x + 1, y, y + 1$. For a such cut, let $N (x, y)$ be the number of trapezoids that have common
points with the line $p$. Define $N (x, y) = \infty$ if there are no trapezoids completely left and
no trapezoids completely right of the line $p$.

\begin{lemma}
Let $G \neq K_n$ be a trapezoid graph. Then
$$
\kappa (G) = \min_{1 \leq x \leq 2n, \ 1 \leq y \leq 2n} N (x, y).
$$
\end{lemma}

\begin{proof}
Let $S$ be a vertex cut of the graph $G$ with the minimum cardinality $\kappa (G)$. The removal of
$S$ disconnects $G$, and consider the component $C$ that contains a trapezoid with the smallest
upper left corner $a [i]$. Let $x$ be the maximum value among upper right corners in the
component~$C$, $x = \max_{i \in C} b [i]$, and let $y$ be the maximum value among lower right
corners in the component~$C$, $y = \max_{i \in C} d [i]$. Since the right border of trapezoids from
$C$ form concave broken line -- the line $p$ that passes through intervals $(x, x + 1)$ and $(y, y
+ 1)$ is one cut of a trapezoid graph. It follows that $\kappa (G) \geq N (x, y)$. The other
inequality follows immediately, and this completes the proof.
\end{proof}

Note that in the above theorem the points $x = 1$ and $x = 2n$, as well as $y = 1$ and $y = 2n$,
can be excluded from the consideration.

Therefore, one can traverse all values $x$ and $y$ and compute the number of trapezoids that have
non-empty intersection with the line $p$ determined by the intervals $(x, x + 1)$ and $(y, y + 1)$.
The important thing is to ensure that there are trapezoids lying entirely to the left and to the
right of the line $p$. This can be easily done in $O (n^2)$.

We will first precompute the leftmost and the rightmost trapezoids for each interval $(x, x + 1)$
from the upper line. Let $leftmost$ be the index of a trapezoid with $b [leftmost] \leq x$ and
minimal lower right corner $d$. Similarly let $rightmost$ be the index of a trapezoid with $a
[rightmost] \geq x + 1$ and maximal lower left corner $c$. If there are no such trapezoids, set the
values of $leftmost$ and $rightmost$ to $-1$. We need additional arrays $index\_up$ and
$index\_bottom$, such that $index\_up [j]$ contains the index of the trapezoid with the left or
right coordinate equal to $j$ on the upper line, and similarly for the bottom line. This can be
done in linear time $O (n)$ as shown in Algorithm~\ref{alg:leftmost} (implementation for the
$rightmost$ array is similar and, thus, omitted).

\LinesNumbered
\begin{algorithm}
    \label{alg:leftmost}
    \KwIn{The trapezoids $T$ and the array $index\_up$.}
    \KwOut{The array $leftmost$.}

    $leftmost [1] = -1$\;
    \For{$j = 2$ \KwTo $2n$}
    {
        $i = index\_up [j]$\;
        $leftmost [j] = leftmost [j - 1]$\;
        \If{$b [i] = j$}
        {
            \If{($leftmost [j] = -1$) {\bf or} ($d [leftmost [j]] > d [i]$)}
            {
                $leftmost [j] = i$;
            }
        }
    }
    \Return $leftmost$\;

    \caption{ Calculating the leftmost trapezoids. }
\end{algorithm}

We will traverse the coordinates on the upper line from $x = 1$ to $x = 2n$, and skip the values
with $leftmost [x] = -1$ or $rightmost [x] = -1$. For each value $y$ between $d [leftmost [x]]$ and
$c [righmost [x]]$, we need to calculate the number of trapezoids $N (x, y)$ that cut the line $p$.
The trapezoid $T [i]$ cuts the line $p$ if
\begin{itemize}
\item it contains the interval $(x, x + 1)$;
\item it is left trapezoid with the lower right corner greater than $y$, i. e. if $b [i] \leq x$ and
$d [i] > y$;
\item it is right trapezoid with the lower left corner less than or equal to $y$, i. e. if $a [i] > x$ and
$c [i] \leq y$.
\end{itemize}

Furthermore, we will maintain the binary array $cut$ of length $n$ that indicates whether the
trapezoid $T [i]$ contains the interval $(x, x + 1)$. In other words, $cut [i] = true$ if and only
if $a [i] \leq x < x + 1 \leq b [i]$. Since no two trapezoids have a common corner, we can update
this array in the constant time by traversing from $x$ to $x + 1$. Therefore, for each trapezoid $T
[i]$ with $cut [i] = true$ we easily check whether this trapezoid is on the left of $x$ or on the
right of~$x$. We can also keep the number of left and right trapezoids in the variables $left$ and
$right$ (see Algorithm \ref{alg:new}).

In order to calculate the minimum value $N (x, y)$ for the fixed $x$ coordinate, we will traverse
$y$ coordinates, and count the number of trapezoids with $cut [i] = false$ that intersect the
line~$p$. The idea is to calculate the cumulative sum by adding $+1$ for each coordinate $c [i]$ of
right trapezoids and by adding $-1$ for each coordinate $d [i]$ of left trapezoids. The starting
value of the cumulative sum is $left$. The number $N (x, y)$ is equal to the number of trapezoids
that contain the interval $(x, x + 1)$ plus the cumulative sum. The pseudo-code of this approach is
given in Algorithm \ref{alg:quadratic}.

\LinesNumbered
\begin{algorithm}
    \label{alg:quadratic}
    \KwIn{The trapezoids $T$ and the arrays $leftmost$, $rightmost$, $index\_bottom$, $cut$ and parameters $x$, $left$, $right$.}
    \KwOut{The minimum value $N (x, y)$ for $1 \leq y \leq 2n$.}

    $sum = left$\;
    $min\_sum = -1$\;
    \For{$y = 1$ \KwTo $c [rightmost [x]]$}
    {
        \If{($y \geq d [leftmost [x]]$) {\bf and} ($y \leq c [righmost [x]]$)}
        {
            \If{($min\_sum = -1$) {\bf or} ($min\_sum > sum$)}
            {
                $min\_sum = sum$\;
            }
        }
        $i = index\_bottom [y]$\;
        \If{($b [i] \leq x$) {\bf and} ($y = d [i]$) {\bf and} ($cut [i] = false$)}
        {
            $sum = sum - 1$\;
        }
        \If{($a [i] > x$) {\bf and} ($y = c [i]$) {\bf and} ($cut [i] = false$)}
        {
            $sum = sum + 1$\;
        }
    }
    \If{$min\_sum > -1$}
    {
        \Return $(n - left - right) + min\_sum$\;
    }
    \Else
    {
        \Return $-1$\;
    }

    \caption{ Calculating the minimum value $N (x, y)$ for the given coordinate $x$. }
\end{algorithm}

\begin{example}
The vertex connectivity of the graph $G$ in Figure \ref{fig:example} is two. For $x = 11$, we have
the following parameters $left = 5$, $right = 2$, $cut [6] = true$ and $cut [i] = false$ for $1
\leq i \leq 8$ and $i \neq 6$. The execution of Algorithm \ref{alg:quadratic} is presented in Table
1.

\begin{table}[!htbp]
\label{tab:example} \centering
\begin{tabular} {l|llllllllllllllll}
\toprule
$i$ & 1 &  2 &  3 &  4  & 5 &  6 &  7 &  8 &  9 &  10 & 11 & 12 & 13 & 14 & 15  & 16 \\
\midrule
$leftmost$ &  -1 & -1 &  1 & 1 &  1 &  1 & 1 & 1 & 1 & 1 & 1 &  1 & 1 &  1 & 1 & 1\\
$rightmost$ &  7 &  7 &  7 &  7 &  7 &  7 &  7 &  7 &  7 & 7 &  7 &  8 &  8 &  -1 &  -1 &  -1 \\
$A$ & 0 & 0 & -1 & -1 & 0 & 0 & 0 & -1 & 0 & -1 & 0 & 1 & -1 &  1 & 0 & 0\\
$sum$ & 5 & 5 & 4 & 3 & 3 & 3 & 3 & 2 & 2 & 1 & 1 & 2 & 1 & 2 & 2 & 2 \\
\bottomrule
\end{tabular}
\caption{Example of the algorithm execution.}
\end{table}

From the above table, we conclude that the minimum value of cumulative sums from $y = 3$ to $y =
13$ equals 1 and it is achieved for $y = 10, 11, 13$. This shows that the vertex connectivity of
$G$ is less than or equal to $1 + (8 - 5 - 2) = 2$.
\end{example}

For the efficient implementation, we will use the modified binary indexed tree data structure. The
array $A$ will correspond to the lower coordinates from $1$ to $2n$. For each trapezoid $T [i]$
with $cut [i] = false$, assign $A [c [i]] = 0$ and $A [d [i]] = -1$ if $T [i]$ is a left trapezoid,
and assign $A [c [i]] = 1$ and $A [d [i]] = 0$ if $T [i]$ is a right trapezoid. In order to find
the minimum value of $N (x, y)$ for $1 \leq y \leq 2n$, we can just return $Calculate (rightmost
[x])$. The only problem is to ensure that there is at least one left trapezoid for each unit
interval $(x, x + 1)$. We can solve this by setting $-n^2 - 1$ for the value $A [d [leftmost [x]]]$
and taking this into account when calculating the minimum values. The number $n^2 + 1$ is big
enough and all partial sums before the index $leftmost [x]$ will be greater than $n$ (we update the
vertex connectivity only if $N (x, y)$ is less than or equal to $n$). At the end we need to handle
the special case -- when $G$ is a complete graph. The pseudo-code of this algorithm is presented
below.

\LinesNumbered
\begin{algorithm}
    \label{alg:new}
    \KwIn{The trapezoids $T$, the arrays $leftmost$, $rightmost$ and $index\_up$, and MBIT data structure.}
    \KwOut{The vertex connectivity number.}

    $k = -1$\;
    $left = 0$\;
    $right = n$\;
    \For{$i = 1$ \KwTo $n$}
    {
        $cut [i] = false$\;
        $Update (c [i], 1)$\;
        $Update (d [i], 0)$\;
    }
    \For{$x = 1$ \KwTo $2n - 1$}
    {
        $i = index\_up [x]$\;
        \If{$a [i] = x$}
        {
            $cut [i] = true$\;
            $right = right - 1$\;
            $Update (c [i], 0)$\;
            $Update (d [i], 0)$\;
        }
        \Else
        {
            $cut [i] = false$\;
            $left = left + 1$\;
            $Update (c [i], 0)$\;
            $Update (d [i], -1)$\;
        }

        \If{($x > 1$) {\bf and} ($leftmost [x - 1] \neq -1$)}
        {
            $Update (d [leftmost [x - 1]], -1)$\;
        }
        \If{($leftmost [x] \neq -1$)}
        {
            $Update (d [leftmost [x]], -n \cdot n-1)$\;
        }
        \If{($rightmost [x] \neq -1$) {\bf and} ($leftmost [x] \neq -1$)}
        {
            $Nxy = Calculate (rightmost [x]) + n \cdot n + (n - right)$\;
            \If{$Nxy \leq n$}
            {
                \If{($k = -1$) {\bf or} ($k > Nxy$)}
                {
                    $k = Nxy$\;
                }
            }
        }
    }

    \If{$k > -1$}
    {
        \Return $k$\;
    }
    \Else(\tcp*[h]{   complete graph})
    {
        \Return $n - 1$\;
    }
    \caption{ Calculating the vertex connectivity of a trapezoid graph. }
\end{algorithm}

Since every trapezoid will be added and removed exactly once from the modified binary indexed tree
data structure, the total time complexity is $O (n \log n)$. This modified data structure with
variable left ends is a novel approach to the best of our knowledge and makes this problem very
interesting.

We conclude this section by summing the results in the following theorem.

\begin{theorem}
The proposed algorithm calculates the vertex connectivity of a trapezoid graph with $n$ vertices in
time $O (n \log n)$ and space $O (n)$.
\end{theorem}

\section{Bipartiteness criteria and tree representation}

In this section we establish local test for bipartiteness of trapezoid graphs.

\begin{theorem}
\label{thm:bipartite} The trapezoid graph $G$ is bipartite if and only if it does not contain a
triangle.
\end{theorem}

\begin{proof}
The first part directly follows from the well-known result: the graph $G$ is bipartite if and only
if it does not contain odd cycles.

Let $G$ be a triangle-free trapezoid graph and let $C = T [1] T [2] \ldots T [k]$ be the smallest
odd cycle contained in $G$ with $k > 3$. It can be easily seen that there are no chords in $C$, i.
e. there are no edges of the form $T [i] T [j]$ with $i > j + 1$ (otherwise we could find smaller
odd cycle). Consider the intersection of the trapezoids $T [1]$ and $T [2]$. If their intersection
is a trapezoid with height equal to the distance of two parallel lines (see the first part of
Figure \ref{fig:bipartite1}), then all trapezoids $T [3], T [4], \ldots, T [k]$ must be on the
right side of $T [1]$ -- which is impossible, since $T [k]$ must have common points with $T [1]$.
Otherwise, the trapezoids $T [1]$ and $T [2]$ have intersection as shown in the second part of
Figure \ref{fig:bipartite1}. Without loss of generality we can assume that the trapezoids $T [3]$
and $T [k]$ are independent and positioned as shown in the figure. In order to connect the
trapezoids $T [3]$ and $T [k]$ by a path of trapezoids $T [4] T [5] \ldots T [k - 1]$ -- some
trapezoids of this path must intersect either $T [1]$ or $T [2]$, which is impossible.

Therefore, the graph $G$ does not contain cycles of odd length and it follows that $G$ is
bipartite.
\end{proof}

\begin{figure}[ht]
  \label{fig:bipartite1}
  \center
  \includegraphics [width = 14cm]{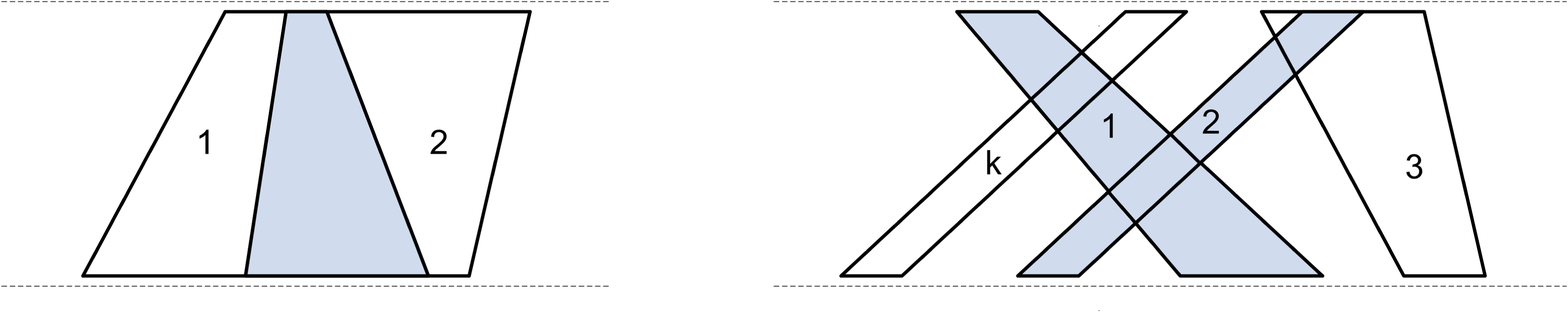}
  \caption { \textit{ Two cases for Theorem \ref{thm:bipartite}. } }
\end{figure}

Note that from the above proof it follows that each cycle of length greater than four contains a
chord (an edge joining two nodes that are not adjacent in the cycle).

\medskip

A caterpillar graph is a tree such that if all pendent vertices and their incident edges are
removed, the remainder of the graph forms a path. Let $C_{n,d} (a_1, a_2, \ldots, a_{d-1})$ be a
caterpillar with $n$ vertices obtained from a path $P_{d+1} = v_0 v_1 \ldots v_{d-1}v_d$ by
attaching $p_i \geq 0$ pendent vertices to the vertex $v_i$, $1\leq i \leq d-1$, where
$n=d+1+\sum_{i=1}^{d-1}p_i$. It can be easily seen that each caterpillar $C_{n,d} (a_1, a_2,
\ldots, a_{d-1})$ has trapezoid representation as triangle-free interval graph \cite{Ju93}.

Assume now that tree $G$ is not a caterpillar and has trapezoid representation. Then it contains a
vertex $w$ with neighbors $v_1, v_2, v_3$, such that each vertex $v_i$ has another neighbor $u_i$
different than $w$, $i = 1, 2, 3$. The trapezoids $T [v_1]$, $T [v_2]$ and $T [v_3]$ corresponding
to the vertices $v_1$, $v_2$ and $v_3$ are independent. Without loss of generality we can assume
the order $T [v_1] \ll T [v_2] \ll T [v_3]$. Since all trapezoids $T [v_1]$, $T [v_2]$ and $T
[v_3]$ intersect the trapezoid $T [w]$, it can be easily seen that all neighbors of $T [v_2]$
(trapezoid $T [u_2]$ in particular) also must intersect $T [w]$. This is a contradiction, and $G$
does not have trapezoid representation. Therefore, we proved the following

\begin{theorem}
\label{thm:trees} A trapezoid graph $G$ represents a tree if and only if it is a caterpillar.
\end{theorem}

\section{Concluding remarks}

In this paper we presented an efficient algorithm for calculating the vertex connectivity number
$\kappa (G)$ of a trapezoid graph. We leave as an open problem to design efficient algorithm for
finding the edge connectivity number $\lambda (G)$ in trapezoid graphs.

The $k$--trapezoid graphs are an extension of trapezoid graphs to higher dimension orders. The
$k$--dimensional box representation $(V, l, u)$ of a graph $G = (V, E)$ consists of mappings $l: V
\rightarrow R^k$ and $u: V \rightarrow R^k$ such that $l [i]$ is the lower and $u [i]$ the upper
corner of a box $box [i]$ where two vertices of the graph are joined by an edge iff their
corresponding boxes are incomparable \cite{FeMuWe97}. If a graph has such a representation, it is a
$k$--trapezoid graph. If we additionally have a weight $w: V \rightarrow R$ on the vertices of $G$
then the $k$--trapezoid graph is weighted. For the case $k = 2$, we have simple trapezoid graphs.

Another generalization are circular trapezoid graphs -- the intersection graphs of circular
trapezoids between two parallel (concentric) circles \cite{Li06}. It seems that the presented
approach can be modified and adapted for calculating the vertex connectivity number of
$k$--trapezoid graphs and circular trapezoid graphs.

\vspace{0.5cm} {\bf Acknowledgement. }  This work was supported by Research Grants 174010 and
174033 of Serbian Ministry of Education and Science.

\end{document}